\documentclass[11pt]{amsart}
\usepackage{amsmath, amsthm, amssymb}
\usepackage{graphicx}
\usepackage{mathbbol}
\usepackage[usenames]{color}
\usepackage{tikz}
\usepackage{alltt}
\usetikzlibrary{shapes}
\usetikzlibrary{plotmarks}

\newtheorem{theorem}{Theorem}[section]

\newtheorem{proposition}[theorem]{Proposition}

\newtheorem{definition}[theorem]{Definition}

\newtheorem{remark}[theorem]{Remark}

\newcommand{\dvol}{\,\mathrm{dvol}}
\newcommand{\R}{\mathbb{R}}
\newcommand{\geucl}{g_{\mathbb{E}}}

\newcommand{\rn}{d_{n}}

\newcommand{\mE}{m_{\mathrm{eff}}}

\begin{document}

\title[]{On a gravitational self-interaction parameter for point-particles}

\author{Iva Stavrov Allen}
\address{Lewis \& Clark College}
\email{istavrov@lclark.edu}

\date{}


\begin{abstract} Relativistic, electrically neutral point-particles can be given mathematical foundation by doing a careful accounting of self-interaction energies; see \cite{NI}. In this paper we examine a self-interaction parameter and present a continuous framework which interpolates between classical and relativistic point-particles of \cite{NI}.
\end{abstract}

\maketitle

\section*{Introduction}

The experience with classical physical theories such as Newtonian gravity makes it tempting to interpret the Schwarzschild space-time as a result of point mass serving as a gravitational source. However, non-linearity of Einstein's equations makes the use of standard (linear, Schwarz) distribution theory (e.g. Dirac delta distributions) difficult. One outstanding question in relativity is whether one can put the above interpretation of Schwarzschild geometry on rigorous mathematical footing. 

One may want to look for such a footing within a class of metrics with distributional curvature. The authors of \cite{Strings} perform an investigation of such metrics and conclude that point-sources are not well-defined even within their wide regularity class:
\begin{quote}
Indeed, it now seems likely that there is in general relativity no mathematical framework whatever for matter sources concentrated on one-dimensional surfaces in space-time.
\end{quote}

A more recent treatment of the subject is given in \cite{HS} where the authors make use of mathematically demanding non-linear distribution theory (Colombeau algebras). The approach is successful in that it presents the Schwarzschild energy-momentum tensor in the form of the Dirac delta distribution. However, since the starting point of \cite{HS} is the Schwarzschild metric itself the approach only makes it \emph{a posteriori} clear that Schwarzschild metric is related to some physical content being concentrated at a point. 

It would be considerably more desirable to have an \emph{a priori} treatment of the subject, that is, a treatment which makes it clear that Schwarzschild metric is \emph{a necessary outcome} of matter becoming concentrated at a point. This is exactly what is achieved in \cite{NI}, albeit in the context of time-symmetric conformally flat initial data. In addition to being mathematically approachable, the method in \cite{NI} makes it manifestly obvious that interaction energies play a central role in the process of concentrating physical matter to a point. In fact, the subject of \cite{NI} is exactly this connection between (self-)interaction and non-linearity of the Hamiltonian constraint (de-facto relativistic Poisson equation). It is argued in \cite{NI} that 
the said connection is captured by a self-interaction parameter, $\alpha$. 

In this note we probe the self-interaction parameter $\alpha$ further. Specifically, we modify the Poisson equation to include this continuous parameter $\alpha$, with $\alpha=0$ corresponding to the classical (Newtonian) Poisson equation, and $\alpha=1$ corresponding to the relativistic Poisson equation of \cite{NI}. Our main result (Theorem \ref{thmbaaaa}) shows that, relative to our modification of the Poisson equation, the matter distributions with self-interaction parameter $\alpha$ result in ``generalized gravitational potentials" of the form of classical gravitational potentials for point-mass $\mE(\alpha)$. We further show that $\mE(\alpha)$ depends continuously on $\alpha$, a feature that the framework in \cite{NI} does not have.

\subsection*{Acknowledgments}
The research is funded by John S. Rogers Science Research Program at Lewis \& Clark College. I.S. is extremely grateful to Noah Benjamin for all the inspiration and all the conversations on the topic, and is saddened by the inability to coauthor this note with him. 

\section{Premilinaries}
\subsection{Notation} 
Throughout our paper we work within the conformal class $\theta^4\geucl$ of the Euclidean metric $\geucl$ and take $\omega=\phi\dvol_{\geucl}$ to be a smooth, compactly supported matter distribution on $\R^3$. We always assume $\phi\ge 0$. The asymptotic conditions which ensure asymptotically Euclidean data are
\begin{equation}\label{asymptotics}
\left|\partial_x^l\!\left(\theta(x)-1\right)\right|=O(|x|^{-l-1}),\ \ |x|\to \infty,\ \ l\ge0.
\end{equation}
Since $R(\theta^4\geucl)=-8\theta^{-5}\Delta_{\geucl}\theta$, the Hamiltonian constraint is equivalent to 
\begin{equation}\label{TheEqn}
\theta \Delta_{\geucl}\theta \dvol_{\geucl}=-4\pi \tfrac{G}{2c^2} \omega.
\end{equation}
We refer to \eqref{TheEqn} as the relativistic Poisson equation, as opposed to the Newtonian Poisson equation $\Delta_{\geucl}\theta \dvol_{\geucl}=-4\pi \tfrac{G}{2c^2} \omega$. 

\subsection{The interaction coupling} We can understand the coupling of $\theta$ to $\Delta_{\geucl}\theta$ in \eqref{TheEqn} as a form of gravitational interaction within the matter $\omega$ itself. This is most easily seen on the example of Brill-Lindquist data \cite{BL}. Specifically, the Brill-Lindquist conformal factor $\theta=1+\frac{G}{2c^2}\sum\frac{m_i}{|x-p_i|}$ corresponds to a configuration of point-particles of mass $m_i$ located at $p_i$. Evaluating the equation \eqref{TheEqn} at such $\theta$ yields to a (mathematically not rigorous) identity
$$\left(1+\frac{G}{2c^2}\sum_j\frac{m_j}{|x-p_j|}\right)\sum_i m_i\delta_{p_i}=\sum_i m_i\delta_{p_i} + \frac{G}{2c^2} \sum_{ij}  \frac{m_i m_j}{|p_i-p_j|}\delta_{p_i}=\omega.$$ 
This identity decomposes the bare mass $\omega$ into the effective mass $\sum_i m_i\delta_{p_i}$, which is expected from the standpoint of the Newtonian  Poisson equation, and a ``new" interaction energy term $\tfrac{G}{2c^2} \sum_{ij}  \frac{m_i m_j}{|p_i-p_j|}\delta_{p_i}$.

\subsection{The self-similarity framework and the summary of \cite{NI}.} The framework of \cite{NI} involves (approximately) self-similar families of distributions $\omega_n$ supported on $B_{\geucl}(0,r_n)$ with $r_n\to 0$. We re-define self-similarity for the purposes of our paper below. The reader should keep in mind that the idea of \cite{NI} is to investigate the limit of geometries $\theta_n^4\geucl$, which are related to $\omega_n$ by means of the relativistic Poisson equation \eqref{TheEqn}. 

\begin{definition}\label{framework}
Let $\Omega$ be a distribution supported on a compact subset of $\R^3$ and let  $\rn=\left(\tfrac{G}{2c^2}\cdot\tfrac{m}{r_n}\right)^{-1}$ with $m := \int_{\R^3}^{}\Omega$. A sequence of distributions $\omega_n$ is self-similar of $\Omega$-type and interactivity $\alpha\in[0,1]$ if for the dilation $\mathcal{H}_{\rn}: x\mapsto \rn\,x$ we have 
\begin{equation}\label{selfsimilaritydefn}
\Omega=(\rn)^\alpha\cdot \mathcal{H}_{\rn}^*\omega_n, \text{\ \ i.e.\ \ }
\omega_n = (\rn)^{-\alpha}\cdot \left(\mathcal{H}_{\rn}\right)_*\Omega.
\end{equation}
\end{definition}

In the case of $\alpha=0$ the sequence of matter distributions $\omega_n$ approaches a multiple of the Dirac delta distribution. The idea behind the factor of $\rn^{-\alpha}$ in \eqref{selfsimilaritydefn} is that it adds extra energy -- the interaction energy present in a construction of point particle initial data. The findings of \cite{NI} are summarized in the following diagram. 

\begin{figure}[h]
\centering
\begin{tikzpicture}[scale=.35]

\draw (-0.5,-6) to (6.25, -6) to (8.25,-4) to (1.5, -4) to (-0.5, -6);

\draw[thick, dashed] (2, -4.5) to [out=-30, in=100] (2.75, -6);
\draw[thick] (2.75, -6) to [out=-80, in=180] (3.75, -6.5) to [out=0, in=-100] (4.75, -6);
\draw[thick, dashed] (4.75, -6) to [out=80, in=-150] (5.75, -4.5);

\node at (2.9, -8.5) {$\theta_n\Delta_{\geucl} \theta_n=-4\pi \tfrac{G}{2c^2} \omega_n$};

\draw (8, -5) to (11, -5);
\draw (10.8, -5.2) to (11, -5) to (10.8, -4.8);

\draw (11.5,-6) to (16.75, -6) to (18.75,-4) to (13.5, -4) to (11.5, -6);

\draw[thick, dashed] (13.9, -6) to [out=60, in=180] (15, -5.25) to [out=0, in=120] (16.1, -6);
\draw[thick] (13.9, -6) to [out=-120, in=90] (13.75, -6.5) to [out=-90, in=180] (15.25, -7.75) to [out=0, in=-90] (16.25, -6.5) to [out=90, in=-60] (16.1, -6);

\draw [fill] (15, -5.25) circle (0.1);

\node at (14.25,-8.5) {$\alpha=0$};

\draw (19.25,-6) to (24.5, -6) to (26.5,-4) to (21.25, -4) to (19.25, -6);

\draw[thin, dashed] (22.5, -4.85) to [out=-30, in=100] (22.95, -5.25) to [out=-90, in=45] (22.75, -5.7);
\draw[thin, dashed] (23.5, -4.85) to [out=-150, in=80]  (23.05, -5.25) to [out=-90, in=135] (23.25, -5.7);

\draw [fill] (23, -5.25) circle (0.1);

\draw (19.5,-6.25) to (24.75, -6.25) to (26.75,-4.25);
\draw[dashed] (26.75,-4.25) to (21.5, -4.25) to (19.5, -6.25);

\node at (22,-8.5) {$0<\alpha<1$};

\draw (26.75,-6) to (31.5, -6) to (33.5,-4) to (28.75, -4) to (26.75, -6);

\draw[thick, dashed] (29, -4.5) to [out=-30, in=90] (29.75, -6);
\draw[thick] (29.75, -6) to [out=-90, in=25] (28.75, -8);
\draw[thick] (31.75, -8) to [out=160, in=-90] (30.6, -6);
\draw[thick, dashed] (30.6, -6) to [out=90, in=-150] (31.45, -4.37);

\draw[thick] (29.8, -6.15) to [out=-20, in=-160] (30.6, -6.15);

\draw (26.5,-9) to (31.5, -9) to (33.5,-7) to (28.5, -7) to (26.5, -9);

\node at (29.5,-8.5) {$\alpha=1$};

\end{tikzpicture}
\caption{Analysis of \cite{NI}.}\label{fig1}
\end{figure}
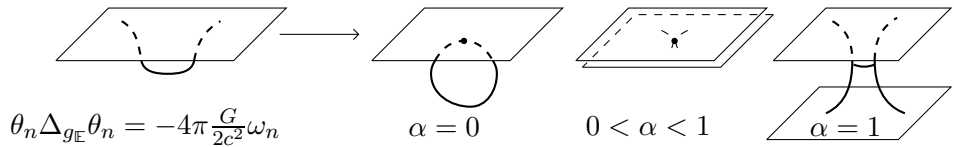

For $\alpha=0$ the ADM mass vanishes in the limit and the contents of $\Omega$ are trapped in ``a bubble". If one adds back an insufficient amount of matter, corresponding to $\alpha\in(0,1)$, the mass still vanishes in the limit. It is only when $\alpha=1$ that one is adding enough energy to obtain  Schwarzschild initial data and non-vanishing mass. The reader should note \emph{the lack of continuity} in $\alpha$ that is present in the mass results of \cite{NI}. 

\section{Generalized Poisson equation}
In this paper we investigate the limits of generalized gravitational potentials $\theta_{\alpha,n}$ obtained from solving the generalized Poisson equation:
\begin{equation}\label{TheAlphaEqn}
\theta^\alpha \Delta_{\geucl}\theta \dvol_{\geucl}=-4\pi \tfrac{G}{2c^2} \omega,\ \ \lim_{x\to \infty} \theta(x)=1
\end{equation}
with $\alpha\in[0,1]$. As above, the coupling of $\theta^\alpha$ to $\Delta_{\geucl}\theta$ models gravitational interactivity. The continuous parameter $\alpha$ marks the departure from the Newtonian Poisson equation ($\alpha=0$) towards its relativistic counterpart (the Brill-Lindquist example; $\alpha=1$). The idea here is to \emph{continuously} transition from the (non-)interactivity of the Newtonian matter toward the full interactivity of the relativistic matter. 

\begin{proposition}\label{NoahThm}
For each $\alpha\in[0,1]$ there exists a unique solution $\theta_\alpha$ of \eqref{TheAlphaEqn}. The family $\theta_\alpha$ is continuous in $\alpha$ in the sense that for all convergent sequences $\alpha_n \to \alpha$ we have convergences $\theta_{\alpha_n}\to \theta_\alpha$ with all derivatives on all compact subsets of $\R^3$.
\end{proposition}

\begin{proof} 
We first establish existence and uniqueness. In the case of $\alpha=0$ there is nothing to show and so we proceed by fixing a value of $\alpha\in(0,1]$. For reasons of notational simplicity we temporarily drop $\alpha$ from the subscript. We implement  the strategy of \cite{NI} which is based on the recursive sequence 
$$ \theta_{m+1}(x):=1+\frac{G}{2c^2}\int_{y\in \R^3} \frac{\omega(y) }{|x-y| \theta_m^\alpha(y)},\ \ \theta_0(x)=1.$$
As in \cite{NI} one proves that the sequences $\theta_{2n}$ and $\theta_{2n+1}$ converge on all compacts to functions $\theta_-$ and $\theta_+$ satisfying $1\le \theta_-\le \theta_+$ and the boundary conditions $\theta_\pm\to 1$. Furthermore, it follows that 
$$\Delta_{\geucl}(k\theta_-+(1-k)\theta_+) =-4\pi \tfrac{G}{2c^2}\phi\cdot  \tfrac{k\theta_-^\alpha+(1-k)\theta_+^\alpha}{\theta_-^\alpha\theta_+^\alpha}$$
for all constants $k$. 

If $\theta_-\neq \theta_+$, i.e $\theta_-^\alpha<\theta_+^\alpha$ somewhere, then for some positive constant $k>1$ the function 
$$\theta_+^\alpha+k(\theta_-^\alpha-\theta_+^\alpha)=k\theta_-^\alpha+(1-k)\theta_+^\alpha$$
achieves a value less than $1$. By taking $k>1$ not too large, we may assume that the said interior minimum value of $k\theta_-^\alpha+(1-k)\theta_+^\alpha$ is positive. Next, we argue that $k\theta_-+(1-k)\theta_+$ for that specific value of $k$ reaches an interior minimum. There would be nothing to show if $k\theta_-+(1-k)\theta_+$ were to turn negative so we assume $k\theta_-+(1-k)\theta_+>0$ on $\R^3$. Since the function $x\mapsto x^\alpha$ is concave down, and since $k>1$, Jensen's Inequality implies 
$$(k\theta_-+(1-k)\theta_+)^\alpha\le k\theta_-^\alpha+(1-k)\theta_+^\alpha$$
over $\R^3$. It follows that the functions $(k\theta_-+(1-k)\theta_+)^\alpha$ and $k\theta_-+(1-k)\theta_+$ reach values -- and thus interior minimum values -- less than $1$. At the point of minimum we have $\Delta_{\geucl}(k\theta_-+(1-k)\theta_+)\ge 0$ while $-4\pi \frac{G}{2c^2}\phi\cdot \frac{k\theta_-^\alpha+(1-k)\theta_+^\alpha}{\theta_-^\alpha\theta_+^\alpha}<0$. This contradiction shows that $\theta_-=\theta_+$, and proves the existence of solutions of \eqref{TheAlphaEqn}. The uniqueness of solutions follows from the Strong Maximum Principle as in \cite{NI}, with very minor modifications to accommodate for the parameter $\alpha$.

We now focus on continuity in the parameter $\alpha$. Fix $0\le \alpha<\beta\le 1$. The difference $\theta_\beta-\theta_\alpha$ satisfies 
$$\Delta_{\geucl}(\theta_\beta-\theta_\alpha)+4\pi\tfrac{G}{2c^2} \tfrac{\phi}{\theta_\alpha^\alpha\theta_\beta^\beta}\left(\theta_\alpha^\alpha-\theta_\beta^\beta\right)=0,$$
which by the Mean Value Theorem becomes
$$
\Delta_{\geucl}(\theta_\beta-\theta_\alpha)-4\pi \alpha \tfrac{G}{2c^2}\tfrac{\phi}{\theta_\alpha^\alpha\theta_\beta^\beta}\theta_*^{\alpha-1}\left(\theta_\beta-\theta_\alpha\right)= 4\pi \tfrac{G}{2c^2} \tfrac{\phi}{\theta_\alpha^\alpha\theta_\beta^\beta}\left(\theta_\beta^\beta-\theta_\beta^\alpha\right)\ge 0.
$$
It follows from the Strong Maximum Principle that $\theta_\beta-\theta_\alpha$ cannot reach a nonnegative interior maximum unless it is a constant. Given that $\theta_\beta-\theta_\alpha$ obeys a Dirichlet boundary condition, we obtain 
\begin{equation}\label{monotone}
1\le \theta_1\le \theta_\beta\le \theta_\alpha\le \theta_0.
\end{equation}

Now suppose that $\alpha_n\to \alpha$; without loss of generality we may assume that the sequence $(\alpha_n)$ is monotone. It follows from \eqref{monotone} that both $(\theta_{\alpha_n})$ and $\left(\frac{\phi}{\theta_{\alpha_n}^{\alpha_n}}\right)$ are bounded in $L^2(K)$ for all compact subsets $K$. By the Interior Elliptic Regularity we see that $(\theta_{\alpha_n})$ is bounded in $H^2(K)$ for all compact subsets $K$. By the Rellich Lemma and the Sobolev embedding we get a subsequential convergence of $(\theta_{\alpha_n})$ to some $\theta\in C^0(K)$. However, monotonicity \eqref{monotone} ensures that the entire sequence 
$(\theta_{\alpha_n})$ converges to $\theta$. The standard bootstrapping argument based on Interior Elliptic Regularity now shows that the convergence to $\theta$ happens in each and every $H^k(K)$. Taking the limit as $n\to \infty$ in the representation formula 
$$\theta_{\alpha_n}(x)=1+\frac{G}{2c^2}\int_y \frac{\omega(y)}{|x-y|\theta_{\alpha_n}(y)^{\alpha_n}}$$
shows that the limit $\theta$ solves \eqref{TheAlphaEqn}. Since the said solutions are unique it must be that $\theta=\theta_\alpha$ and our proof is complete. 
\end{proof}

\begin{remark} It may be tempting to try to prove Proposition \ref{NoahThm} by the method of continuity / Implicit Function Theorem. The difficulty in that approach lies in the need to use Fredholm theory for Laplace operators on weighted function spaces. Unfortunately, the weights needed to get such a proof off the ground are outside the Fredholm range. 
\end{remark}

To prove our main result we need the following extension of Proposition \ref{NoahThm}. The reader should note two things about the extension: one is the existence of the solution of the generalized Poisson equation with the zero boundary condition (important, as it appears in the statement of our main result), while the other is its continuity in $\alpha$. 

\begin{proposition}\label{AngryIvaThm}
For each $a\in[0,1]$ and each $\alpha\in[0,1]$ there exists a unique solution $\theta_{a,\alpha}$ of 
\begin{equation}\label{TheAlphaEqn-modified}
\theta_{a,\alpha}^\alpha \Delta_{\geucl}\theta_{a,\alpha} \dvol_{\geucl}=-4\pi \tfrac{G}{2c^2} \omega,\ \ \lim_{x\to \infty} \theta_{a,\alpha}(x)=a.
\end{equation}
The family $\theta_{a,\alpha}$ depends continuously on $(a,\alpha)$ in the sense that for all convergent sequences $(a_n,\alpha_n)\to (a,\alpha)$ in the permissible range we have convergences $\theta_{a_n,\alpha_n}\to \theta_{a,\alpha}$ with all derivatives on all compact subsets of $\R^3$.
\end{proposition}

\begin{proof}
The uniqueness of solutions of \eqref{TheAlphaEqn-modified} follows by the same Strong Maximum Principle argument as in the proof of Proposition \ref{NoahThm}. The existence of solutions of \eqref{TheAlphaEqn-modified} in the cases when $a\neq 0$ is a consequence of Proposition \ref{NoahThm} because a function $\theta_{a,\alpha}$ serves as a solution of \eqref{TheAlphaEqn-modified} if and only if the function $\frac{\theta_{a,\alpha}}{a}$ serves as a solution of \eqref{TheAlphaEqn} with $\omega$ replaced by $\frac{1}{a^{1+\alpha}}\omega$. The existence of solutions of \eqref{TheAlphaEqn-modified} for $a=0$ is established later on in this proof. 

Temporarily fix some $0<a<b\le 1$ and a value of $\alpha\in[0,1]$. Consider solutions $\theta_{a,\alpha}$ and $\theta_{b,\alpha}$ of \eqref{TheAlphaEqn-modified}. By the Mean Value Theorem the difference $\theta_{a,\alpha}-\theta_{b,\alpha}$ satisfies
$$\Delta_{\geucl}(\theta_{a,\alpha}-\theta_{b,\alpha})-\frac{4\pi\alpha\phi}{\theta_*^{\alpha+1}}\left(\theta_{a,\alpha}-\theta_{b,\alpha}\right)=0$$
for some positive function $\theta_*$. 
We see from the Strong Maximum Principle that $\theta_{a,\alpha}-\theta_{b,\alpha}$ cannot reach a nonnegative interior maximum unless it is a constant. Since $\theta_{a,\alpha}-\theta_{b,\alpha}\to a-b<0$ we arrive at $\theta_{a,\alpha}<\theta_{b,\alpha}$. This further gives $\theta_{a,\alpha}^\alpha< \theta_{b,\alpha}^\alpha$ which, when combined with the representation formula, yields  
$$\begin{aligned}
\theta_{a,\alpha}(x)=&a+\frac{G}{2c^2}\int \frac{\omega(y)}{|x-y|\theta_{a,\alpha}^\alpha(y)}\\
\ge & (a-b)+b+\frac{G}{2c^2}\int \frac{\omega(y)}{|x-y|\theta_{b,\alpha}^\alpha(y)}=(a-b)+\theta_{b,\alpha}(x).
\end{aligned}$$
Overall, we have 
\begin{equation}\label{cty-a}
0\le \theta_{b,\alpha}-\theta_{a,\alpha}\le b-a.
\end{equation}
Th estimate \eqref{cty-a} and the monotonicty formula \eqref{monotone} provide a lower bound $\theta_{a,\alpha}\ge \theta_{1,\alpha}-1\ge \theta_{1,1}-1$, valid for all $a\in (0,1]$ and all $\alpha\in[0,1]$.
For future purposes we note that  
\begin{equation}\label{lowerbound}
\theta_{a,\alpha}\big{|}_{\mathrm{supp}(\omega)}\ge c:=\min_{\mathrm{supp}(\omega)} (\theta_{1,1}-1)=\min_{\mathrm{supp}(\omega)}\frac{G}{2c^2}\int \frac{\omega(y)}{|x-y|\theta_{1,1}(y)} >0
\end{equation} 
for all $a\in (0,1]$ and all $\alpha\in[0,1]$.

Next, temporarily fix $0\le \alpha<\beta\le 1$ and a value $a\in(0,1]$. Consider the solution $\theta_{\mathrm{aux}}$ of the problem 
$$\theta_{\mathrm{aux}}^\alpha \Delta_{\geucl} \theta_{\mathrm{aux}} =-4\pi \tfrac{G}{2c^2} \tfrac{1}{a^{1+\beta}}\phi,\ \ \theta_{\mathrm{aux}}\to 1,$$
and note that 
$$\left(a^{\frac{\beta-\alpha}{1+\alpha}}\theta_{\mathrm{aux}}\right)^\alpha
\Delta_{\geucl}\left(a^{\frac{\beta-\alpha}{1+\alpha}}\theta_{\mathrm{aux}}\right)=-4\pi \tfrac{G}{2c^2} \tfrac{1}{a^{1+\alpha}}\phi,\ \ a^{\frac{\beta-\alpha}{1+\alpha}}\theta_{\mathrm{aux}}\to a^{\frac{\beta-\alpha}{1+\alpha}}\le 1.$$
The monotonicity formula \eqref{monotone} implies 
$$\tfrac{\theta_{a,\beta}}{a}\le \theta_{\mathrm{aux}}$$
while the inequality \eqref{cty-a} gives 
$$a^{\frac{\beta-\alpha}{1+\alpha}}\theta_{\mathrm{aux}}\le \tfrac{\theta_{a,\alpha}}{a}.$$ 
Since $a^{\frac{\alpha-\beta}{1+\alpha}}\le a^{\alpha-\beta}$ due to $a\in (0,1]$, we obtain 
\begin{equation}\label{newmonotonicity}
\theta_{a,\beta}\le a^{\alpha-\beta}\theta_{a,\alpha},\ \ \text{i.e.}\ \ a^\beta\theta_{a,\beta}\le a^\alpha\theta_{a,\alpha}.
\end{equation}

At this stage we may repeat the argument from the end of the proof of Proposition \ref{NoahThm}, with monotonicity formula \eqref{newmonotonicity} replacing \eqref{monotone}. The conclusion is the continuity of the sequence $\theta_{a,\alpha}$ in $\alpha$ for each fixed $a\in(0,1]$. 

Finally, fix $\alpha\in[0,1]$ and consider a sequence $a_n\to a$. By \eqref{cty-a} we see that the sequence $(\theta_{a_n,\alpha})$ is \emph{uniformly} Cauchy in $C^0(K)$ for each compact set $K$. Combining with \eqref{lowerbound} we obtain that both $(\theta_{a_n,\alpha})$ and $\left(\frac{\phi}{\theta_{a_n,\alpha}^{\alpha}}\right)$ are uniformly Cauchy in $L^2(K)$ for each compact $K$. The Interior Elliptic Regularity and a standard bootstrapping argument show that $(\theta_{a_n,\alpha})$ is uniformly Cauchy in each $H^k(K)$. We see from the representation formula that the limit function $\theta$ satisfies 
\begin{equation}\label{zerocase}
\theta(x)=a+\frac{G}{2c^2}\int \frac{\omega(y)}{|x-y|\theta^\alpha(y)}.
\end{equation}
For $a>0$ the identity \eqref{zerocase} shows $\theta=\theta_{a,\alpha}$ i.e. continuity of $\theta_{a,\alpha}$ in $a$. In the case of $a=0$ we use \eqref{zerocase} to first establish the existence of solutions of \eqref{TheAlphaEqn-modified} at $a=0$. The continuity at $a=0$ then follows from the uniformity of the limit $\theta_{a_n,\alpha}\to \theta_{0,\alpha}$. 
\end{proof}

\section{The main result}

The main result of our paper can be summarized by saying that for a fixed $\alpha\in[0,1]$ the generalized gravitational potentials $\theta_{\alpha,n}$ arising from matter distributions of Definition \ref{framework} via the generalized Poisson equation \eqref{TheAlphaEqn} converge to a classical point-source gravitational potential 
$$\theta_{\mathrm{Schw},\alpha}(z)=1+\frac{G}{2c^2}\cdot\frac{\mE(\alpha)}{|z|}.$$
The effective mass $\mE(\alpha)$ of the limit is given in terms of the solution $\Theta_{\alpha}$ of  
\begin{equation}\label{defnThetaF}
\Theta_\alpha^\alpha \Delta_{\geucl}\Theta_\alpha\dvol_{\geucl}=-4\pi \tfrac{G}{2c^2} \Omega,\ \ \Theta_\alpha\to 0.
\end{equation}
and is equal to 
$$\mE(\alpha):=\int_{y}\frac{\Omega(y)}{\Theta_\alpha^\alpha(y)}.$$
Note that $\mE(\alpha)$ varies continuously in $\alpha$ due to Proposition \ref{AngryIvaThm}.

\begin{theorem}\label{thmbaaaa}
The functions $\theta_{\alpha,n}$ converge to $\theta_{\mathrm{Schw},\alpha}$ uniformly with all derivatives on all compact subsets of $\R^3\setminus\{0\}$. Furthermore, the effective mass $\mE(\alpha)$ of $\theta_{\mathrm{Schw},\alpha}$ changes continuously in $\alpha$. 
\end{theorem}

\begin{proof}
We introduce an auxiliary sequence of functions 
$$\Theta_{\alpha,n}(x)=d_n\mathcal{H}_{d_n}^* \theta_{\alpha,n}(x)=d_n\theta_{\alpha,n}(d_nx)$$
and note the change in the boundary condition: $\Theta_{\alpha,n}\to d_n$. Pulling back the equation \eqref{TheAlphaEqn} along $\mathcal{H}_{d_n}$ gives  
$$\Theta_{\alpha,n}^\alpha \Delta_{\geucl}\Theta_{\alpha,n} \dvol_{\geucl}=-4\pi \tfrac{G}{2c^2} \Omega.$$
The Representation Formula produces 
$$\Theta_{\alpha,n}(x)=d_n+\frac{G}{2c^2}\int_y \frac{\Omega(y)}{|x-y|\Theta_{\alpha,n}^\alpha(y)},$$
with the integration going only over $y\in \mathrm{supp}(\Omega)$. 
For $x \notin \mathrm{supp}(\Omega)$ the expression 
$\frac{1}{|x-y|}$ can be expanded using a power series, which converges when $\frac{|y|}{|x|}$ is substantially small. Inserting this expansion into the representation formula for $\Theta_{\alpha,n}$ yields 
$$\Theta_{\alpha,n}(x)=d_n+\frac{1}{|x|}\frac{G}{2c^2}\int_{y}\frac{\Omega(y)}{\Theta_{\alpha,n}^\alpha(y)}+\sum_{l=1}^\infty \frac{C_l\left(x\right)}{|x|^{2l+1}} \frac{G}{2c^2}\int_{y}\frac{P_l(y) \Omega(y)}{\Theta_{\alpha,n}^\alpha(y)}$$
where $C_l$ and $P_l$ are homogeneous polynomials of degree $l$. Overall, it follows that for some constant $C$ and for all $|z|\geq d_n C$ we have 
$$\theta_{\alpha,n}(z)=1+\frac{1}{|z|}\frac{G}{2c^2}\left(\int_{y}^{}\frac{\Omega(y)}{\Theta_{\alpha,n}^\alpha(y)}\right)+\sum_{l=1}^{\infty}\frac{C_l(z)}{|z|^{2l+1}}\frac{G}{2c^2} \left(\int_{y}\frac{P_l(y) \Omega(y)}{\Theta_{\alpha,n}^\alpha(y)}\right)d_n^{l}.$$

By Proposition \ref{AngryIvaThm} we see that (on all compacts, and with all derivatives) we have $\Theta_{\alpha,n}\to \Theta_\alpha$ where $\Theta_\alpha$ is as in  \eqref{defnThetaF}. In particular, we have that the coefficients $\int_{y}^{}\frac{\Omega(y)}{\Theta_{\alpha,n}^\alpha(y)}$ and $\int_{y}\frac{P_l(y) \Omega(y)}{\Theta_{\alpha,n}^\alpha(y)}$ converge to $\int_{y}^{}\frac{\Omega(y)}{\Theta_\alpha^\alpha(y)}$ and $\int_{y}\frac{P_l(y) \Omega(y)}{\Theta_\alpha^\alpha(y)}$. The claim of our theorem is now immediate from boundedness of $|z|$ (away from $0$ and  $\infty$), and the fact that $d_n\to 0$. 
\end{proof}

\section{Conclusion and directions for further research}
Any treatment of point-particles (including the one of \cite{NI}) should be taken with a grain of salt until one can prove a theorem which realizes a dust cloud (or the like) as a cumulative effect (integral?) of point-particles. That such a theorem is possible within the framework of \cite{NI} is made plausible by a result announced in \cite{Banff}, where a time-symmetric, conformally flat, dust initial data distribution is seen as a pointed intrinsic flat limit of Brill-Lindquist initial data. We plan to investigate the possibility of this kind of a theorem in the future.

\bibliographystyle{plain}
\bibliography{2014}

\end{document}